%% file: main.tex
\documentclass[11pt]{article}

\input{preamble}

\usepackage{thmtools}
\usepackage{thm-restate}
\usepackage{arydshln}
\usepackage[linesnumbered,ruled,vlined]{algorithm2e}
\definecolor{darkred}{rgb}{0.5,0,0}
\definecolor{darkgreen}{rgb}{0,0.35,0}
\definecolor{darkblue}{rgb}{0,0,0.55}
\hypersetup{
    pdftitle={Symbolic determinant identity testing and non-commutative ranks of matrix Lie algebras},
    pdfauthor={G\'abor Ivanyos, Tushant Mittal, Youming Qiao},
    bookmarksnumbered=true,     
    bookmarksopen=true,         
    bookmarksopenlevel=1,       
    colorlinks=true,                     
    pdfpagemode=UseOutlines,    %
    pdfpagelayout=OneColumn,
    pdfstartview=FitH,linkcolor=darkblue,filecolor=darkred,citecolor=darkgreen,urlcolor=darkred,pagebackref
	}

\usepackage[capitalise,nameinlink]{cleveref}

\title{
Symbolic determinant identity testing 
and \\
non-commutative ranks of matrix Lie algebras
}

\author{
G\'abor Ivanyos\thanks{Institute for Computer Science and Control, 
E\"otv\"os Lor\'and Research Network (ELKH),
Budapest, Hungary 
({\texttt Gabor.Ivanyos@sztaki.hu}).} 
\and 
Tushant Mittal\thanks{Department of Computer Science, the University of Chicago, 
Chicago, USA. 
({\texttt tushant@uchicago.edu})}
\and 
Youming Qiao\thanks{Centre for Quantum Software and Information, 
 University of Technology Sydney, Australia ({\texttt Youming.Qiao@uts.edu.au}) 
 }
}

\date{\today}

\begin{document}

\maketitle

\begin{abstract}
One approach to make progress on the symbolic determinant identity 
testing 
(SDIT) problem is to study the structure of singular matrix spaces. After 
settling the non-commutative rank problem (Garg--Gurvits--Oliveira--Wigderson, 
\emph{Found. Comput. Math.} 
2020; Ivanyos--Qiao--Subrahmanyam, 
\emph{Comput. 
Complex.} 2018), a 
natural next step is to understand singular matrix spaces whose non-commutative 
rank is full. At present, examples of such matrix spaces are mostly sporadic, 
so it is desirable to discover them in a more systematic way.

In this paper, we make a step towards this direction, by studying the family of matrix spaces that are closed 
under the commutator operation, that is, matrix Lie algebras. On the one hand, we demonstrate that matrix 
Lie algebras over the complex number field give rise to singular matrix spaces 
with full non-commutative ranks. On the other hand, we show that SDIT of such 
spaces can be decided in deterministic polynomial time. Moreover, we give a 
characterization for the matrix Lie algebras to yield a matrix space possessing 
singularity certificates as studied by Lov\'asz 
(\emph{B. Braz. Math. Soc.}, 
1989) and Raz and Wigderson (\emph{Building Bridges II}, 2019).
\end{abstract}

\section{Introduction}

\subsection{Background and motivations}

\paragraph{Matrix spaces.} Let $\F$ be a field. We use 
$\M(\ell\times n, \F)$ to denote the linear space of 
$\ell\times n$ matrices over $\F$, and let $\M(n, \F):=\M(n\times n, \F)$. The 
general linear group of degree $n$ over $\F$ is denoted by $\GL(n, \F)$. A subspace $\cB$ 
of $\M(\ell\times n, \F)$ is called a \emph{matrix space}, denoted by $\cB\leq 
\M(\ell\times n, \F)$. Given $B_1, \dots, B_m\in \M(n, \F)$, $\langle B_1, \dots, B_m\rangle$ is 
the linear span of the $B_i$'s. In algorithms, $\cB\leq \M(n, \F)$ is naturally represented by a 
linear basis $B_1, \dots, B_m\in\M(n, \F)$. 

Two major algorithmic problems about matrix spaces are as follows. 

\paragraph{The symbolic determinant identity testing problem.} For $\cB\leq\M(n, 
\F)$, let $\crk(\cB)$ be 
the maximum rank over all matrices in 
$\cB$. We say that $\cB$ is \emph{singular}, if $\crk(\cB)<n$. To decide whether 
$\cB$ is singular is known as the \emph{symbolic determinant identity testing} 
(SDIT) problem. The \emph{maximum 
rank problem} for $\cB$ then asks 
to compute $\crk(\cB)$. 
The complexity of SDIT depends on the 
underlying field $\F$. When $|\F|=O(1)$, SDIT is $\coNP$-complete \cite{BFS99}. 
When 
$|\F|=\Omega(n)$, by 
the polynomial identity testing lemma \cite{Sch80,Zip79}, SDIT admits a 
randomized efficient algorithm. To present a deterministic polynomial-time 
algorithm for SDIT is a major open problem in computational complexity, as that 
would imply 
strong circuit lower bounds %
by the seminal work of Kabanets and 
Impagliazzo\cite{KI04}. 

\paragraph{The shrunk subspace problem.} For $\cB\leq\M(n, \F)$ and $U\leq \F^n$, 
the image of $U$ under $\cB$ is 
$\cB(U):=\{B u : B\in \cB, u\in 
U\}$. We say that $U$ is 
a \emph{shrunk subspace} of $\cB$, if $\dim(U)>\dim(\cB(U))$. The  
problem of deciding whether $\cB$ admits a shrunk subspace is \emph{the shrunk 
subspace problem}. 
 The \emph{non-commutative 
rank problem}\footnote{The name ``non-commutative'' rank comes from a natural 
connection between matrix spaces and symbolic matrices over skew fields; see 
\cite{GGOW20,IQS18} 
for details.} 
\cite{GGOW20,IQS18} asks to 
compute $\ncrk(\cB):= \max\{\dim(U)-\dim(\cB(U)) : U\leq\F^n\}$. That is, $\cB$ admits a shrunk subspace if and only if its 
non-commutative rank is 
not full, i.e. $<n$. 
This problem is known for its 
connections 
to invariant theory, linear algebra, graph theory, and quantum information. 
Major progress in the past few years lead to deterministic efficient algorithms 
for the shrunk subspace problem, one by Garg, Gurvits, Oliveira, and Wigderson 
over fields of characteristic $0$ 
\cite{GGOW20}, and the other by Ivanyos, Qiao, and Subrahmanyam over any field
\cite{IQS17,IQS18}.  

\paragraph{Motivations of our investigation.} 
Note that 
if a matrix space admits a 
shrunk subspace, then it has to be singular. However, there exist singular 
matrix spaces
 without shrunk subspaces. After 
settling the shrunk subspace problem \cite{GGOW20,IQS18}, such matrix spaces form 
a bottleneck for further progress on SDIT.  Moreover,  ideas from these works are 
not expected to directly generalize as it was recently shown that the space of 
singular matrices cannot be seen as the null-cone of any reductive group action 
\cite{MW21}

Two classical examples of such subspaces are 
as follows \cite{Lov89}.
\begin{example}\label{ex:sing}
\begin{enumerate}
\item 
Let $\Lambda(n, \F)$ be the 
linear space of alternating matrices, 
namely 
matrices satisfying $\forall v\in \F^n$, $v^tAv=0$.\footnote{When 
$\F$ is of 
characteristic not $2$, a matrix is alternating if and only if it is 
skew-symmetric.}  When $n$ is odd, 
$\Lambda(n, \F)$ is singular, as every 
alternating
matrix is of 
even rank. Furthermore, it is easy to verify that $\Lambda(n, \F)$ does not admit 
shrunk subspaces.
\item  Let $C_1, \dots, C_n\in \Lambda(n, \F)$, and let $\cC\leq \M(n, \F)$ 
consist of 
all the 
matrices of the form $[C_1v, C_2v, \dots, C_nv]$, over $v\in \F^n$. As $C_i$'s are 
alternating, we have 
$$v^t[C_1v, 
C_2v, \dots, 
C_nv]=[v^tC_1v, v^tC_2v, \dots, 
v^tC_nv]=0,$$ 
so $\cC$ is singular. In \cite{EH88}, it is shown that when $n=4$, 
certain choices of 
$C_i$ ensure that $\cC$ does not have shrunk subspaces.
\end{enumerate}
\end{example}

While there are further examples in \cite{AW83,DM17}, the above two examples (and 
their certain subspaces) have been studied most in theoretical computer science 
and combinatorics, such as by Lov\'asz 
\cite{Lov89} and Raz and Wigderson \cite{RW19}, due 
to their connections to matroids and graph rigidity. 

As far as we see from the above, examples of singular matrix spaces without shrunk 
subspaces in the literature are sporadic.  
Therefore, it is desirable to discover more singular matrix 
spaces without shrunk subspaces, hopefully in a more systematic way. This is the 
main motivation of this present article. 

\paragraph{Overview of our main results.} Noting that the linear space of 
skew-symmetric matrices is closed under the 
commutator bracket, 
we set out to study matrix Lie algebras. 
Our main results can be summarized as 
follows. 
\begin{itemize}
\item First, we show that matrix Lie algebras over $\C$ gives rise to a family of 
singular 
matrix spaces without shrunk subspaces. This result, partly inspired by 
\cite{draisma_maximal}, %
vastly 
generalizes the linear spaces 
of skew-symmetric matrices. 
\item Second, we 
present a deterministic polynomial-time 
algorithm to solve SDIT for matrix Lie algebras over $\C$. This 
algorithm heavily relies on the structural theory of, and algorithms for, Lie 
algebras. 
\item Third, we examine when matrix Lie algebras are of the form in 
\cref{ex:sing} (2) as above, giving representation-theoretic criteria for 
such matrix Lie 
algebras. 
\end{itemize}

In the rest of this introduction, we detail our results. 

\subsection{Our results}\label{subsec:result}

Recall that $\cB\leq \M(n, \F)$ is a \emph{matrix Lie algebra}, if $\cB$ 
is closed under the commutator bracket, i.e. for any $A, B\in \cB$, $[A, 
B]:=AB-BA\in \cB$. 

We have striven %
to make this introduction as self-contained 
as possible. 
In an effort to make this article accessible to wider 
audience,
we summarize notions and results on Lie algebras and representations relevant 
to this paper in Appendices~\ref{app:basic}, \ref{app:corresp}, and 
\ref{app:cartan}.

\paragraph{Two results and a message.} We first study shrunk subspaces of matrix 
Lie algebras over $\C$. To state our results, we need the following notions.

Given a matrix space 
$\cB\leq\M(n, \F)$, $U\leq\F^n$ is an \emph{invariant subspace} of $\cB$, if 
for any $B\in \cB$, $B(U)\subseteq U$. We say that $\cB$ is \emph{irreducible}, if 
the 
only invariant subspaces of $\cB$ are $0$ and $\F^n$. The above notions naturally 
apply to matrix Lie algebras. The matrix space 
$\cB=0\leq\M(1, 
\F)$ is called the \emph{trivial irreducible} matrix Lie algebra. 

In general, let $\cB\leq\M(n, \F)$ be a matrix Lie algebra. Then there exists 
$A\in\GL(n, \F)$, such that $A^{-1}\cB A$ is of block upper-triangular form, and 
each block on the diagonal defines an irreducible matrix Lie algebra, called a 
\emph{composition factor} of $\cB$. Such an $A$ defines a chain of subspaces, 
called a 
\emph{composition series} of the matrix Lie algebra $\cB$. By the Jordan-H\"older 
theorem, the isomorphic types of the 
composition factors are the same for different composition series. 

We then have the following criteria for the existence of shrunk subspaces of 
matrix Lie 
algebras over $\C$.
\begin{restatable}{theorem}{Shrunk}\label{thm:shrunk}
Let $\cB\leq\M(n, \C)$ be a non-trivial irreducible matrix Lie algebra. Then 
$\cB$ does not have a shrunk subspace. 

Let $\cB\leq\M(n, \C)$ be a matrix Lie algebra. Then $\cB$ has a shrunk subspace, 
if and only if one of its composition factors is the trivial matrix 
Lie algebra. 
\end{restatable}
The proof of \cref{thm:shrunk} for the irreducible case makes use of 
the connection of Lie algebras and Lie groups as summarized in 
\cref{app:corresp}. Going from the irreducible to the general case, we prove some 
basic properties of shrunk subspaces 
which may be of independent interest in \cref{sec:ncrk}. 

After we proved \cref{thm:shrunk}, we learnt that 
Derksen and Makam independently proved it using a different approach via 
representation theory of Lie algebras \cite{DM21}.

We then present a deterministic polynomial-time algorithm to solve SDIT for matrix 
Lie algebras over $\C$. Our model of computation over $\C$ will be explained in 
\cref{sec:maxrk}.
\begin{theorem}\label{thm:max_rank}
Let $\cB\leq\M(n, \C)$ be a matrix Lie algebra. Then there is a deterministic 
polynomial-time algorithm to solve the symbolic determinant identity testing 
problem for $\cB$.
\end{theorem}
We believe that the strategy for the algorithm in \cref{thm:max_rank} is 
interesting. It rests on the key observation 
that the maximum rank of $\cB$ is equal to 
the maximum rank of a Cartan subalgebra of $\cB$. %
(We collect the 
notions and results on Cartan algebras relevant to this paper in 
\cref{app:cartan}.) We then resort to the algorithm 
computing a Cartan subalgebra by de Graaf, Ivanyos and R\'onyai \cite{cartan} to 
get one. As Cartan subalgebras are upper-triangularisable, an SDIT algorithm can 
be devised easily.

Theorems~\ref{thm:shrunk} and~\ref{thm:max_rank} together bring out \emph{the main 
message in this paper}: we identify 
non-trivial irreducible matrix Lie algebras over $\C$ as an interesting families 
of matrix spaces, as (1) they do not admit shrunk subspaces, and (2) SDIT for such 
spaces can be solved in deterministic polynomial time. 

To see that 
matrix Lie algebras do form an interesting family for the maximum rank problem, we list some examples. 

\begin{example}\label{ex:irr}
\begin{enumerate}
\item Note that $\Lambda(n, 
\F)$ is closed under the commutator bracket.
Indeed, $\Lambda(n, 
\F)$ together with the commutator bracket is well-known as the orthogonal Lie 
algebra, and it is easy to see that it is irreducible. 
\item Representations of abstract Lie algebras give rise to matrix Lie algebras. 
For example, let $\fksl(n, \C)$ be the special 
linear Lie algebra, i.e, the Lie algebra of all $n \times n$ 
complex matrices with trace 0.  Let $E_{i,j}$ be the elementary matrix with the only non-zero entry being $1$ in the $(i,j)^{th}$ entry. A linear basis of $\fksl(n, \C)$ consists of $E_{i, j}$, $i\neq j$. Consider for any fixed $d$, the vector 
space 
$V$ spanned by all degree $dn$ monomials in the variables $\{x_1, \cdots, x_n\}$. 
Then, the representation is defined as $\rho(E_{ij})(x_1^{e_1}\cdots x_n^{e_n}) = 
x_i \frac{\partial(x_1^{e_1}\cdots x_n^{e_n} )}{\partial x_j}$. This gives rise to 
an irreducible matrix Lie algebra in $\M(\binom{dn+n-1}{n-1}, \C)$.
\end{enumerate}
\end{example}

One may wonder whether irreducible matrix Lie algebras encompass singular and 
non-singular matrix spaces. To see this, note that $\Lambda(n, \F)$ (as defined in 
\cref{ex:irr}) can be singular (for odd $n$)
or non-singular (for even $n$). In fact, there is a 
representation-theoretic explanation for the maximum rank of certain irreducible 
matrix Lie algebras via weight spaces (\cref{fact:singular}) as already 
observed by Draisma \cite{Dra06}, from which it is 
evident that irreducible matrix Lie algebras can be singular or non-singular. 

\paragraph{Other singularity witnesses and matrix Lie algebras.} 

After Theorem~\ref{thm:shrunk} and~\ref{thm:max_rank}, we study further 
properties of matrix Lie algebras related to singularity as 
follows. 
Let $\cB=\langle B_1, \dots, B_m\rangle\leq \M(n, \F)$ be a matrix space. Let 
$x_1, \dots, x_m$ be a set of commutative variables. Then $B=x_1B_1+\dots +x_mB_m$ 
is a matrix of linear forms in $x_i$'s. When $\F$ is large enough, the
singularity of $\cB$ is equivalent to that of $B$ over the function field. 
Viewing $B$ as a matrix over the rational 
function field $\F(x_1, \dots, x_n)$, its kernel is spanned by vectors whose 
entries are polynomials. Let $v\in \F[x_1, \dots, x_m]^n$ be in $\ker(B)$. By 
splitting $v$ according to degrees if necessary, we can assume that $v$ is 
\emph{homogeneous}, i.e. each 
component of $v$ is homogeneous of degree $d$. 

We are interested in those vectors in the kernel whose entries are linear forms. 
This is also partly motivated by understanding witnesses for singularity of matrix 
spaces, as by \cite{KI04}, putting SDIT in $\NP\cap \coNP$ already implies strong 
circuit lower bounds. Suppose $B$ admits $v\in \ker(B)$ whose components are 
homogeneous degree-$d$ polynomials. Then, ignoring bit complexities, $v$ is a singularity witness of $\cB$ of size $O(m^{d}\times n)$, and the existence of a certificate of degree $d$ can be checked and found in time $O(m^d\times n)$, by writing a linear system in $O(m^d\times n)$ variables.

Let $v_1, \dots, v_m\in \F^n$, and 
$v=x_1v_1+\dots+x_mv_m$ be a vector of linear forms. We say that $v$ is a 
(left homogeneous) \emph{linear kernel vector} of $B$, if each entry of $v^tB$ is 
the zero polynomial. Similarly, $v$ is a right homogeneous linear kernel vector, if each 
entry of $Bv$ is the zero polynomial.

Clearly, whether such a nonzero $v$ exists does not depend on the choice of bases. Indeed, we can give a 
basis-free definition of a linear kernel vector for a matrix space $\cB  \leq \M(n,\F)$ as a non-zero linear map $\psi: \cB \to \F^n$ such that for each $A$,  $\psi(A)^t A = 0$.  

Matrix spaces with linear kernel vectors have appeared in papers by Lov\'asz \cite{Lov89} 
and Raz and Wigderson \cite{RW19}. To see this, note that matrix spaces with  linear kernel vectors can be 
constructed from alternating matrices as exhibited in 
\cref{ex:sing} (2).

One approach for Lie algebras to yield matrix spaces with linear kernel vectors is 
through adjoint representations. 

Recall that, given a Lie algebra 
$[-,-]:\fkg\times \fkg\to\fkg$, the adjoint 
representation of $\fkg$ is $\ad:\fkg \to \fkgl(\fkg)$ 
defined as $\ad_x(y) = [x,y]$ for $x, y\in \fkg$. The image of $\ad$ is a matrix 
space 
$\cA\leq \M(d, \F)$ where $d=\dim(\fkg)$. As the Lie bracket $[,]$ is alternating, 
$\cA$ admits a 
linear kernel vector by the 
construction in \cref{ex:sing} (2).

Our next theorem characterizes Lie algebra representations with linear kernel 
vectors.
(We collect some basic notions of Lie algebra representations 
relevant to this 
paper in \cref{app:basic}.)
Since we are concerned with matrix spaces which are images of Lie algebra 
representations, i.e. $\cB = \rho(\fkg)$ where $\rho$ is a representation of the 
Lie algebra $\fkg$, 
we can assume without generality that $\rho$ is faithful.

\begin{restatable}{theorem}{LinKer}\label{thm:lin_ker}
Let $\cB$ be the image of a faithful irreducible representation $\phi$ of a semisimple Lie 
algebra $\fkg$ over algebraically closed fields of characteristic not $2$ or $3$. 
Then $\cB$ admits a linear kernel vector if and only if $\cB$ is trivial, or 
$\fkg$ is simple and $\phi$ is isomorphic to the adjoint representation.
\end{restatable}

\subsection{Open questions.}
Several questions can be asked after this work. First, can we identify more 
families of singular matrix spaces without shrunk subspaces? Second, our algorithm 
for SDIT of matrix Lie algebras heavily relies on the structure theory of Lie 
algebras and works over $\C$. It will be interesting to devise an alternative 
algorithm that is of a different nature, and works for matrix Lie algebras over 
fields of positive characteristics. Third, characterize 
those representations of non-semisimple Lie algebras with linear kernel vectors.

\paragraph{The structure of the paper.} In \cref{sec:ncrk} we prove some 
results on shrunk subspaces that will be useful to prove \cref{thm:shrunk}. 
In \cref{sec:shrunk} we prove \cref{thm:shrunk}. In 
\cref{sec:maxrk} we prove \cref{thm:max_rank}. In 
\cref{sec:lin_ker} we prove \cref{thm:lin_ker}.

\section{On shrunk subspaces of matrix spaces}\label{sec:ncrk}

In this section we present some basic results and properties regarding shrunk 
subspaces and non-commutative ranks of matrix spaces.

\subsection{Canonical shrunk subspaces}

Let $\cB\leq \M(n, \F)$. For a subspace $U$ of $\F^n$ define
$\sd_\cB(U)$ as the difference $\dim(U)-\dim(\cB(U))$. Thus 
$\sd_\cB(U)$ is positive for a shrunk subspace $U$ and negative
if $\cB$ expands $U$. We then have the following.
\begin{lemma}
\label{lem:shrunkdim-supermod}
The function $\sd_\cB$ is supermodular. More specifically,
if $U_1$ and $U_2$ are two subspaces
of $\F^n$, then,
\begin{equation}
\label{eq:shrunkdim-supermod}
\sd_\cB(U_1\cap U_2)+\sd_\cB(\langle U_1\cup U_2\rangle)\geq
\sd_\cB(U_1)+\sd_\cB(U_2).
\end{equation}
\end{lemma}

\begin{proof}
By modularity of the dimension,
we have
\[\dim(U_1\cap U_2)+\dim(\langle U_1\cup U_2\rangle)=\dim(U_1)+\dim(U_2)\]
and
\[\dim(\cB(U_1)\cap \cB(U_2))+\dim(\langle \cB(U_1)\cup \cB(U_2)\rangle)=
\dim(\cB(U_1))+\dim(\cB(U_2)).\]
The second equality, using also that $\cB(\langle U_1\cup U_2\rangle)=
\langle \cB(U_1)\cup \cB(U_2)\rangle$ and
$\cB(U_1\cap U_2)\leq \cB(U_1)\cap \cB(U_2)$,
gives, 
\[\dim(\cB(U_1\cap U_2))+\dim(\cB(\langle U_1\cup U_2\rangle))\leq
\dim(\cB(U_1))+\dim(\cB(U_2)).\]
Subtracting the last inequality from the first equality
gives (\ref{eq:shrunkdim-supermod}).
\end{proof}

\begin{proposition}\label{obs:canonical}
Let $\cB\leq \M(n, \F)$. Suppose $\ncrk(\cB)=n-c$ for $c>0$. Then there exists a 
unique 
subspace $U\leq \F^n$ of the smallest dimension satisfying 
$\dim(U)-\dim(\cB(U))=c$, and there exists a 
unique 
subspace $U'\leq \F^n$ of the largest dimension such that 
$\dim(U')-\dim(\cB(U'))=c$.
\end{proposition}
\begin{proof}
We use the supermodular function $\sd_\cB$ defined in
\cref{lem:shrunkdim-supermod}. Let $U_1$ and $U_2$ be 
subspaces with $\sd_\cB(U_i)=c$. Then \cref{lem:shrunkdim-supermod} gives
$\sd_\cB(U_1\cap U_2)+\sd_\cB(\langle U_1\cup U_2\rangle)\geq 2c$.
On the other hand, by the definition of the noncommutative rank,  $\sd_\cB(U_1\cap 
U_2)\leq c$ and
$\sd_\cB(\langle U_1\cup U_2\rangle)\leq c$. It follows that all the three 
inequalities are in fact
equalities. Thus the intersection as well as the span of
all the subspaces $U$ with $\sd_\cB(U)=c$ also have this property.
\end{proof}

By \cref{obs:canonical}, in the case $\ncrk(\cB)=n-c$ for $c>0$, we 
shall refer to the subspace $U$ of the smallest dimension satisfying 
$\dim(U)-\dim(\cB(U))=c$ as the 
(lower) %
\emph{canonical shrunk subspace}. The algorithm 
from \cite{IQS17,IQS18} actually computes the canonical shrunk subspace. 

A natural group action on matrix spaces is as follows. Let $G=\GL(n, \F)\times 
\GL(n, \F)$. Then
$(A, C)\in G$ sends $\cB\leq \M(n, \F)$ to $A\cB 
C^{-1}=\{ABC^{-1} : B\in \cB\}$. The stabilizer group of this action on $\cB$  is 
denoted as 
$\Stab(\cB)=\{(A, C)\in G : A\cB C^{-1}=\cB\}$. 
We then have the following proposition.
\begin{proposition}\label{obs:normalizer}
Let $\cB\leq \M(n, \F)$. Suppose $\ncrk(\cB)=n-c$\footnote{Recall that $\ncrk(\cB):= \max\{\dim(U)-\dim(\cB(U)) : U\leq\F^n\}$. } for $c>0$. Then for 
$\forall (A,C)\in \Stab(\cB)$, the canonical shrunk 
subspace $U$ is invariant under $C$, i.e., $C(U)=U$.
\end{proposition}
\begin{proof}
From the definition of $\Stab(\cB)$, we have 
$A\cB C^{-1} = \cB$ and thus, $A\cB = \cB C$.
Consider the subspace $A(U)$. Then, $\cB (C(U)) =(\cB C) (U) = (A\cB ) (U) = A 
(\cB (U)) $.
Since $A, C \in \GL(n,\F)$, $\dim(C(U)) = \dim(U)$ and $\dim(A (\cB (U))) = \dim( 
\cB(U))$. It follows that $C(U)$ is also a  $c$-shrunk subspace of the same 
dimension as $U$. 
We then conclude that $C(U)=U$ by \cref{obs:canonical}.
\end{proof}

\subsection{Shrunk subspaces of block upper-triangular matrix 
spaces}\label{subsec:diag}

Consider the following situation. Suppose $\cB\leq \M(n, \F)$
satisfies that any $B\in \cB$ is in the block upper-triangular form, i.e. 
\[
B=\begin{bmatrix}
C_1 & D_{1,2} & \dots & D_{1,d} \\
0 & C_2 & \dots & D_{2,d} \\
\vdots & \vdots & \ddots & \vdots \\
0 & 0 & \dots & C_d
\end{bmatrix},
\]
where $C_i$ is of size $n_i\times n_i$. Let 
\[\cC_i=\langle C_i\in \M(n_i, \F) : C_i 
\text{ appears as the }i\text{th diagonal block of some } B\in \cB\rangle.\]

\begin{lemma}
\label{lem:blockdiag}
Let $\cB\leq \M(n, \F)$, and let $V\leq \F^n$
such that $\cB(V)\leq V$. If there exists a shrunk subspace
for $\cB$, then there also exist one which is either included
in $V$ or contains $V$. 
\end{lemma}

\begin{proof}
Assume that $V$ itself is not a shrunk subspace. Then $\cB(V)=V$.
Let $U$ be a shrunk subspace of $\cB$. By \cref{lem:shrunkdim-supermod},
we have $\sd_\cB(V\cap U)+\sd_\cB(\langle V\cup U\rangle)\geq
\sd_\cB(V)+\sd_\cB(U)=0+\sd_\cB(U)>0$. Thus either 
$\sd_\cB(V\cap U)$ or
$\sd_\cB(\langle V\cup U\rangle)$ must be positive.
\end{proof}

The following proposition characterizes the existence of shrunk subspaces in block 
upper-triangular matrix spaces. 
\begin{proposition}\label{prop:blockd}
Let $\cB\leq \M(n, \F)$ and $\cC_i\leq \M(n_i, \F)$, 
$i\in[d]$, as above.  Then $\cB$ has a shrunk subspace if and only if there exists 
$i\in[d]$ such that $\cC_i$ has a shrunk subspace.
\end{proposition}
\begin{proof}
The if direction can be verified easily. For the only if direction, we induct on 
$d$. When $d=1$, this is clear. Suppose this holds for $d< k$. Consider $d=k$, and 
suppose $\cB$ admits a shrunk subspace. Let $V\leq \F^n$ be the subspace spanned 
by those standard basis vectors $e_{n_1+1}, e_{n_1+2}, \dots, e_n$. We then have 
two cases. 
\begin{enumerate}
\item There exists a shrunk subspace $W\leq V$. In this case, by the induction 
hypothesis, there exists $i\in\{2, \dots, n\}$ such that $\cC_i$ has a shrunk 
subspace.
\item There are no shrunk subspaces $W\leq V$. Then by \cref{lem:blockdiag}, 
there exists a shrunk subspace $W$ such that $W>V$. Then by considering $W/V$, we 
obtain a shrunk subspace for $\cC_1$.
\end{enumerate}
This concludes the proof of \cref{prop:blockd}.
\end{proof}

\section{Shrunk subspaces of matrix Lie algebras over $\C$}\label{sec:shrunk}

In this section, we will give a characterization of those matrix Lie 
algebras over $\C$ with shrunk subspaces, proving \cref{thm:shrunk}. The 
main reason for working over $\C$ is to make use of the connections between Lie algebras and Lie groups as described in \cref{app:corresp}. 

We will first give such a characterization for irreducible matrix Lie algebras. 
The general case then follows by combining this 
with the results in 
\cref{subsec:diag}.

The key to understanding the irreducible case lies in the following lemma; for notions such as matrix exponentiation and derivation, cf. \cref{app:corresp}.

\begin{lemma}[{\cite[Proposition 4.5 (1)]{hall}}] \label{exp}
Let $\cB\leq\M(n, \C)$ be an irreducible matrix Lie algebra. %
Let $W \leq \C^n$ and $M \in \cB$. 
If $e^{tM}(W) \leq W$ for all $t\in\R$, %
then $M(W) 
\leq W$ .
\end{lemma}
\begin{proof}
Take any $w\in W$. 
Note that $ \frac{d(e^{tM})}{dt}(w) = (Me^{tM})(w)=M(e^{tM}(w))$, and 
$\frac{d(e^{tM})}{dt} = \lim_{h \rightarrow t} \frac{e^{hM}-e^{tM}}{h-t}$. So at 
$t=0$, 
we have $M(w)=\lim_{h \rightarrow 
0} \frac{e^{hM}(w)-w}{h}$. Since
$e^{tM}(w) \in W$ for all $t\in \R$, $\frac{e^{hM}(w)-w}{h}$ lies in $W$ for any 
$h$, and so 
does the 
limit which is $M(w)$. 
\end{proof}

We will also need the following result. 
\begin{lemma}\label{cor:stab}
Given a matrix Lie algebra $\cB\leq\M(n, \C)$, we have that $\forall t\in \R$ 
and $M \in \cB$, $e^{tM}\cB e^{-tM}=\cB$.
\end{lemma}
\begin{proof}
By the connection between Lie groups and Lie algebras (cf. 
\cref{thm:subalgebra}), there exists some Lie group $G$ whose associated 
Lie algebra is $\cB$. This implies that for any $M \in 
\cB$, $e^{tM} \in G$. Then by the fact that the conjugation of $g\in G$ stabilizes 
$\cB$ (cf. \cref{expstab}), we have 
$e^{tM}\cB e^{-tM}=\cB$.
\end{proof}

We are now ready to prove \cref{thm:shrunk}.

%
\Shrunk*
\begin{proof}
We first handle the irreducible case. 

For the sake of contradiction, suppose $\cB$ has a shrunk subspace. Then let 
$V=\C^n$, and let $U 
\leq V$ 
be the canonical shrunk subspace of $\cB$.
By \cref{cor:stab}, for any $M\in \cB$, we have that $(e^{tM}, e^{tM}) 
\in 
\Stab(\cB)$. 
By
\cref{obs:normalizer}, $U$ is invariant under $e^{tM}$. By 
\cref{exp}, $U$ is an invariant subspace of $\cB$. 

Since $\cB$ is irreducible 
as a matrix Lie algebra, the 
only 
invariant subspaces are 
$0$ and $V$. Since $U$ is a shrunk subspace, it cannot be $0$. If $U=V$, then 
$\cB(V)$ is a proper subspace of $V$. 
If $\cB(V)$ is non-zero, then $\cB(\cB(V))\leq\cB(V)$. This implies that $\cB(V)$ 
is a proper and non-zero invariant subspace of $\cB$, which is impossible as $\cB$ 
is irreducible. 
It follows that $U=V$ and $\cB(V)=0$. 
In this case, $V$ must be of dimension $1$, as any non-zero proper subspace of $V$ 
is an 
invariant subspace. It follows that
$\cB$ has to be the trivial 
matrix Lie algebra. We then arrive at the desired contradiction. 

The general case follows from the irreducible case as shown above, and 
\cref{prop:blockd}. 
\end{proof}

\section{SDIT for matrix Lie algebras over $\C$}\label{sec:maxrk}

In this section, we present a deterministic polynomial-time algorithm that solves 
SDIT for matrix Lie algebras 
over $\C$, proving \cref{thm:max_rank}.

The basic idea is to 
realize that $\cB$ is singular if and only if every Cartan subalgebra of $\cB$ is 
singular. Furthermore, a Cartan subalgebra is nilpotent, so in particular it is 
solvable. It follows, by Lie's theorem (\cref{thm:Lie}), that a Cartan 
subalgebra of a matrix Lie algebra is 
upper-triangularisable by the 
conjugation action. A key task here is to compute a Cartan subalgebra of $\cB$. 
This 
problem has been solved by de Graaf, Ivanyos, and R\'onyai in \cite{cartan}. 

\paragraph{Computation model over $\C$.} We adopt the following computation model 
over $\C$, in consistent with that in \cite{cartan}. That is,  we 
assume the input matrices are over a number field $\mathbb{E}$. Therefore 
$\mathbb{E}$ is a 
finite-dimensional algebra over $\Q$. If $\dim_\Q(\mathbb{E})=d$, then 
$\mathbb{E}$ is the 
extension of $\F$ by a single generating element $\alpha$, so $\mathbb{E}$ can be 
represented by the minimal polynomial of $\alpha$ over $\F$, together with an 
isolating rectangle for $\alpha$ in the case of $\C$. 

\subsection{Cartan subalgebras.} We collect notions and results on Cartan 
subalgebras useful to us in 
\cref{app:cartan}. 
Here, we recall the following. Let $\fkg$ be a Lie algebra. A 
subalgebra $\fkh\subseteq \fkg$ is a \emph{Cartan subalgebra}, if it is nilpotent 
and self-normalizing. 

In \cite{cartan}, de Graaf, Ivanyos, and R\'onyai studied the problem of computing 
Cartan subalgebras. We state the following version of their main result in our 
context as follows. For a more precise statement, see 
\cref{thm:algoregular}.
\begin{theorem}[{\cite[Theorem 5.8]{cartan}}]\label{thm:cartan}
Let $\cB\leq\M(n, \C)$ be a matrix Lie algebra. Then there exists a deterministic 
polynomial-time algorithm that computes a linear basis of a Cartan subalgebra 
$\cA$ of $\cB$.
\end{theorem}

\subsubsection{Maximum ranks of Cartan subalgebras.} The key lemma that supports 
our algorithm is the following. 

\begin{lemma}\label{lem:key}
Let $\cB\leq \M(n, \C)$ be a matrix Lie algebra. Let $\cA\leq \cB$ be a Cartan 
subalgebra. Then,  $\crk(\cB)=\crk(\cA)$.
\end{lemma}
\begin{proof}
We shall utilise two results about Cartan subalgebras; for 
details see \cref{app:cartan}. 

First, let $\fkg$ be a Lie algebra over a 
large enough field. Then there 
exists a set of generic\footnote{This 
means that after identifying $\fkg$ with $\F^{\dim(\fkg)}$, these elements form 
a Zariski open set.} elements $R\subseteq \fkg$, 
such that for any $x\in R$, the Fitting null component of $\ad_x$, 
$\FitZero(\ad_x) 
=\{y\in \fkg : \exists m>0, \ad_x^m(y) = 
0\}$, is a Cartan subalgebra. For a precise statement, see \cref{reg}.

Second, let $\cB$ be a matrix Lie algebra over $\C$. Then for any two Cartan 
subalgebras $\cA$, $\cA'$ of $\cB$, they are conjugate, namely there exists $T\in 
\GL(n, \C)$ such that $T\cA 
T^{-1}=\cA'$.
For a precise statement, see \cref{conj}.

By the first result, in particular by the fact that elements in $R$ are generic, 
there exists a matrix $C\in \cB$ of rank $\crk(\cB)$, such 
that $\cC:=\FitZero(\ad_C)$ is a Cartan subalgebra. Noting that $C\in \cC$, 
$\crk(\cC)=\crk(\cB)$. By the second result, for any Cartan subalgebra $\cA$ of 
$\cB$, $\cA$ and $\cC$ are conjugate, which implies that 
$\crk(\cA)=\crk(\cC)=\crk(\cB)$. 
\end{proof}

\subsection{Upper-triangularisable matrix spaces.} Let $\cB\leq\M(n, \F)$. We say 
that $\cB$ is upper-triangularisable, if there exists $S, T\in \GL(n, \F)$, such 
that for any $B\in \cB$, $S\cB T$ is upper-triangular. Upper-triangularisable 
matrix spaces are of interest to us, because solvable matrix Lie algebras can be 
made simultaneously upper-triangular via conjugation by Lie's theorem 
(\cref{thm:Lie}).

If a matrix space is 
upper-triangularisable, then we can decide if $\cB$ is singular in a 
black-box fashion, as its singularity is completely determined by the 
diagonals of the resulting upper-triangular matrix space. The following lemma is 
well-known and we include a proof for 
completeness. 
\begin{lemma}\label{lem:pisigma}
Let $n, k\in \N$. Let $\F$ be a field such that $|\F| >  (k-1)n$. There exists a 
deterministic algorithm that outputs in time $\poly(n, k)$ a set 
$\cH \subseteq \F^k$, 
such that any non-zero $k$-variate degree-$n$ polynomial, which is a product of linear 
forms,  evaluates to a non-zero value on at least one 
point in $\cH$.
\end{lemma}
\begin{proof}
 Let $\ell_1, \dots, \ell_n$ be $n$ non-zero linear forms in $k$ variables. We can 
 also identify them as vectors in $\F^k$ by taking their coefficients.
 Fix  a  subset 
 $S \subseteq \F$ of size $(k-1)n +1$.  Let $\cH = \{(1, \alpha, 
 \cdots, \alpha^{k-1}) \;|\; \alpha \in S  \}$. 
This is clearly a set of size 
 $(k-1)n +1$.  
 
 We claim that any non-zero linear form $\ell_i$ vanishes on at most 
 $k-1$ points in $\cH$. This is because if it vanishes on $k$ points, we have 
 $A\ell_i = 0$ where $A$ is the Vandermonde matrix corresponding to those $k$ 
 points.  This is impossible because the Vandermonde matrix is invertible 
 and $\ell_i$ is non-zero. 
 
 It follows that there is at least one point in $\cH$ such that every $\ell_i$ has a non-zero evaluation
at this point. This concludes the proof. 

\end{proof}

\subsection{The algorithm.} Given the above preparation, we present the following 
algorithm for computing the commutative rank of a matrix Lie algebra.

\begin{description}
\item[Input:] $\cB=\langle B_1, \dots, B_m\rangle\leq \M(n, \C)$, such that $\cB$ 
is a matrix Lie algebra.
\item[Output:] ``Singular'' if $\cB$ is singular, and ``Non-singular'' otherwise.
\item[Algorithm:]
\begin{enumerate}
\item Use \cref{thm:cartan} to obtain $\cC=\langle C_1, \dots, C_k\rangle 
\leq\cB$, such that $\cC$ is a Cartan subalgebra of $\cB$. 
\item Use \cref{lem:pisigma} to obtain $H\subseteq\C^k$, $|H|=(k-1)n+1$.
\item For any $(\alpha_1, \dots, \alpha_k)\in H$, if $\sum_{i\in[k]}\alpha_iC_i$ 
is non-singular, return ``Non-singular''.
\item Return ``Singular''.
\end{enumerate}
\end{description}

The above algorithm clearly runs in polynomial time. The correctness of the above 
algorithm follows from 
Lemmas~\cref{lem:key,lem:pisigma}, as well as Lie's theorem on solvable 
Lie algebras (\cref{thm:Lie}).
This concludes the proof of \cref{thm:max_rank}.

\begin{remark}
We do not solve the maximum rank problem for matrix Lie algebras in general. While 
the 
maximum rank problem for matrix Lie algebras reduces to the maximum rank problem 
for upper-triangularisable matrix spaces through Cartan subalgebras, to compute 
the maximum rank for the latter deterministically seems difficult. This is because 
the maximum rank problem for upper-triangularisable matrix spaces
is as difficult as the general SDIT problem, an observation already in 
\cite{IKQS15}. 

There is one case where we do solve the maximum rank problem, that is, when the matrix 
Lie algebra over $\C$ is semisimple. In this case, Cartan subalgebras are 
diagonalizable \cite[Theorem in Chapter 6.4]{Hum12}. 
Therefore, in the above algorithm we can output the maximum rank over 
$\sum_{i\in[k]}\alpha_iC_i$ where $(\alpha_1, \dots, \alpha_k)\in H$ as the 
maximum rank of $\cB$. 
\end{remark}

\section{Linear kernel vectors of matrix Lie algebras}\label{sec:lin_ker}

The goal of this section is to study existence of linear kernel vectors for matrix 
spaces arising from representations of Lie algebras.

Let $\fkg$ be a Lie algebra and $(\rho,V)$ be a representation of $\fkg$, where 
$V\cong \F^n$. Let $\cB = \rho(\fkg)\leq \M(n, \F)$.

First, note that $\cB$ admits 
a common kernel vector\footnote{That is $v\in \F^n$ such that for any $B\in \cB$, 
$Bv=0$.} if and only if $(\rho,V)$ has a trivial subrepresentation. We view this as a 
degenerate case, so in the following we shall 
mainly consider representations without 
trivial subrepresentations.

By the basis-free definition of linear kernel vectors in 
\cref{subsec:result}, $\cB = \rho(\fkg)$ has a linear kernel vector if we 
have 
a linear map $\beta: \rho(\fkg) \to V$ such that  $\rho(x) \beta(\rho(x)) = 0$.  
For our purposes, it will be more convenient to work with a map from 
$\fkg$ itself to $V$. This leads us to define 
that for a representation $(\rho, V)$,  
the linear map
$\psi: \fkg \to V$ is a \emph{linear kernel vector} if
\begin{equation}
\label{eq:linker-def}
\rho(x)\psi(x)=0 
\end{equation}
for every $x\in \fkg$. We further assume that $\psi$ is not identically zero. 

\begin{remark}
The definition of linear kernel vectors above is a generalization
that allows for possibly more linear kernel vectors. 
This is because a linear kernel vector 
$\beta$ yields a generalized one by taking 
$\psi = \beta \circ \rho$.  
However, when $\rho$ is not injective,  
we can have many more generalized maps. For example,  
for a trivial representation $(0,V)$,  
$\beta$ has to be $0$ but any linear map 
from $\fkg$ to $V$ is a generalized linear kernel vector.
\end{remark}

Applying \cref{eq:linker-def} to $x$, $y$ and $x+y$ one obtains  for 
every $x,y\in \fkg$,
\begin{equation}
\label{eq:linker-cross}
\rho(x)\psi(y)+\rho(y)\psi(x)=0
\end{equation}

Since $[x,x] = 0$ for the adjoint representation, the identity map of $\fkg$ and 
its scalar multiples are generalized linear kernel vectors.

Assume that $\psi:\fkg\rightarrow V$ is a linear kernel vector for $(\rho,V)$. 
Let $(\rho',V')$ be another representation of $\fkg$. \
Then, if $\phi:V\rightarrow V'$ is a non-zero 
linear map such that $\phi\circ \rho=\rho'\circ \phi$ (that is, $\phi$ is a 
homomorphism between the two representations) then $\phi\circ \psi$ is a linear 
kernel vector for $(\rho',V')$. Indeed, 
$\rho'(x)\phi(\psi(x))=\phi(\rho(x)(\psi(x))=\phi(0)=0$.  

 Our aim is to show that for many of Lie algebras $\fkg$, unless the 
 representation $(\rho,V)$ includes a trivial subrepresentation, every linear 
 kernel vector $\psi:\fkg\rightarrow V$ can be obtained as the composition of the 
 adjoint representation and a homomorphism.

\begin{theorem}\label{thm:lin_main}
Let $\fkg$ be a semisimple Lie algebra $\fkg$ over 
an algebraically closed
 field $\F$ of 
characteristic not $2$ or $3$.  Assume that that the trivial representation is not 
a subrepresentation of the representation $(\rho,V)$ of $\fkg$.  Then any linear 
kernel vector $\psi$ defines a homomorphism  $\psi: (ad,  \fkg) \to (\rho,V)$ 
i.e.,   for every $x,y \in \fkg$,  
\begin{equation}
\label{eq:linker-hom}
\psi([x,y])-\rho(x)\psi(y)=0.
\end{equation}
\end{theorem}

We defer the proof of \cref{thm:lin_main} in \cref{subsec:lin_main}. 
We now derive a corollary of \cref{thm:lin_main} and use it to prove 
\cref{thm:lin_ker}.

\begin{corollary}\label{cor:lin_ker}
Let $\fkg, (\rho,V)$ satisfy the condition in \cref{thm:lin_main}.  
Then,  $(\rho,V) \cong \oplus_i (\ad,\fkg_i) \oplus (\rho',V')$ 
where $\fkg_i$ are 
not necessarily disjoint or distinct
quotient algebras of $\fkg$,  and $(\rho', V')$ has no linear kernel vectors.  
\end{corollary}
\begin{proof}
Let $\psi$ be a linear kernel vector $(\rho, V)$. 
Let $V_1 = \im(\psi) \cong \fkg/\ker{\psi} =: \fkg_1$. Then $V_1$ is invariant 
under $\rho$ as for any $x\in \fkg,  \psi(y) \in V_1$,  $\rho(x)\psi(y) = 
\psi([x,y]) \in V_1$.  Therefore,   $(\rho,V) \cong (\ad, \fkg_1) \oplus (\rho', 
V' )$ by the semisimplicity of $\fkg$. We can then repeat till we 
no longer have linear kernel vectors.
\end{proof}

\LinKer*
\begin{proof}
When $\cB$ is trivial it clearly has a linear vector kernel.  So assume it is 
not.  Applying \cref{cor:lin_ker} in the case of $\rho$ being 
irreducible,  we get 
$(\rho,V) \cong (\ad, \fkg')$ for some quotient algebra $\fkg'$ of $\fkg$. Since 
$\rho$ is faithful,  we must have $\fkg = 
\fkg_i$. By definition,  subrepresentations of the adjoint representation are the 
same as ideals.  Thus irreducibility of the adjoint representation implies that 
$\fkg$ is simple.
\end{proof}

\begin{remark}
\begin{enumerate}
\item 
\cref{thm:lin_main} and 
\cref{thm:lin_ker} also 
hold over sufficiently large perfect fields as a semisimple Lie
algebra over a perfect field remains semisimple over
the algebraic closure of such fields. However, passing over to the closure need 
not preserve semisimplicity in general and thus, the current proof of 
\cref{thm:gen2} does not work for any sufficiently large field.

\item Initially we proved a fact equivalent to
(\cref{eq:linker-hom}) for irreducible representations of
classical Lie algebras using Chevalley bases. In an attempt to simplify
the proof by reducing to certain subalgebras and 
taking trivial subrepresentations into account, we discovered
relevance of equalities (\ref{eq:linker-key-x}) and
(\ref{eq:linker-key-y}) below. We include the previous proof in \cref{app:old}, 
as some ideas and techniques there may be useful for future references. 
\end{enumerate}
\end{remark}

\subsection{Proof of \cref{thm:lin_main}}\label{subsec:lin_main}

To prove \cref{thm:lin_main}, we need the following preparations.
 
\begin{proposition}
\label{lem:linker-key}
Let $\psi:\fkg\rightarrow V$ be a linear kernel vector for 
the representation $(\rho,V)$ of a Lie algebra $\fkg$. Then,
\begin{equation}
\label{eq:linker-key-x}
\rho(x)\big(\psi([x,y])-\rho(x)\psi(y)\big)=0
\end{equation}
and
\begin{equation}
\label{eq:linker-key-y}
\rho(y)\big(\psi([x,y])-\rho(x)\psi(y)\big)=0
\end{equation}
Consequently, the difference
$\psi([x,y])-\rho(x)\psi(y)$ is annihilated
by $\rho(z)$ for every element $z$ of the Lie
subalgebra of $\fkg$ generated by $x$ and $y$.
\end{proposition}

\begin{proof}
\begin{align*}
\rho(x)\psi([x,y])&=-\rho([x,y])\psi(x) & \text{Applying (\ref{eq:linker-cross}) to}\; [x,y], x \\
&=-\rho(x)\rho(y)\psi(x)+\rho(y)\rho(x)\psi(x)& (\rho,V) \text{ is a representation}  \\
&= -\rho(x)\rho(y)\psi(x)& \text{Applying (\ref{eq:linker-def}) to} \; x\\
&=\rho(x)\rho(x)\psi(y) &\text{Applying (\ref{eq:linker-cross}) to} \; x,y
\end{align*}

Similarly (\cref{eq:linker-key-y}) follows because 
\begin{eqnarray*}
\rho(y)\psi([x,y])=-\rho([x,y])\psi(y)=
-\rho(x)\rho(y)\psi(y)+\rho(y)\rho(x)\psi(y)=
\rho(y)\rho(x)\psi(y).\qedhere
\end{eqnarray*}
\end{proof}

The following 
statement %
follows from standard density arguments but we provide a proof in \cref{app:density}.

\begin{proposition}[\cref{lem:span2}]
\label{lem:span}
Let $\fkg$ be an $m$-dimensional Lie algebra over
a large enough field %
$\F$, such that there exist two
elements $x_0,y_0\in\F$ that
generate $\fkg$. Then there are elements $x_i,y_i\in \fkg$,
($i=1,\ldots,m^2$) such that $x_i$ and $y_i$ generate
$\fkg$ for every $i$ and that $x_i\otimes y_i$ span $\fkg\otimes \fkg$.
\end{proposition}

The following lemma is a major step to prove \cref{thm:lin_main}.

\begin{lemma}\label{thm:linear_kernel}
Let $\psi:\fkg\rightarrow V$ be a linear kernel vector for 
the representation $(\rho,V)$ of a Lie algebra $\fkg$ over a large 
enough %
field 
$\F$.
Suppose $\fkg$ can be generated by two elements, and 
the trivial
representation is not a subrepresentation of 
the representation $(\rho,V)$ of $\fkg$. Then for every $x,y\in \fkg$,
\begin{equation*}
\psi([x,y])-\rho(x)\psi(y)=0.
\end{equation*}
\end{lemma}

\begin{proof}
By standard arguments, it is sufficient
 to prove the theorem for the special case when
$\F$ is algebraically closed. We assume that.
By \cref{lem:span}, we have $\{(x_i,y_i)\}$ that each generate $\fkg$ 
as a 
Lie algebra and collectively linearly span $\fkg\otimes\fkg$.  For each $i $,   by 
\cref{lem:linker-key},
$\rho(z)\big(\psi([x_i,y_i])-\rho(x_i)\psi(y_i)\big)=0$ for every $z \in \fkg$.
This equality is trilinear and therefore it holds for every $z\otimes x\otimes y$ 
for
$(z,x,y)\in \fkg\times ( \mathrm{span}_i \{(x_i\otimes y_i)\})$. 
By \cref{lem:span},  $\mathrm{span}_i \{(x_i\otimes y_i)\} = 
\fkg\otimes\fkg$ 
and thus, $\rho(z)\big(\psi([x,y])-\rho(x)\psi(y)\big)$ is identically zero
on $\fkg^{\otimes 3}$.  Now for every fixed $(x,y)$ the vector 
$\psi([x,y])-\rho(x)\psi(y)$ 
is annihilated by all of $\rho(\fkg)$. It follows that the vector must be zero, as 
otherwise it would span a trivial subrepresentation. 
\end{proof}

To 
deduce \cref{thm:lin_main} from 
\cref{thm:linear_kernel}, we need a result for the number of generators of 
certain Lie algebras. 
We recall some classical results on this topic. First, Kuranishi 
gave a simple proof that over characteristic 0, there 
exists two elements that generate any semisimple Lie algebra \cite[Thm. 
6]{Kuranishi51}. The proof of the 
statement works directly in positive characteristic ($> 3$) for sum of 
\emph{classical simple} lie algebras, i.e.  those obtained from a Chevalley basis. 
This was extended in  \cite{Bois09} to other kinds of simple Lie algebras over 
positive characteristic ($> 3$). This can be extended to the semisimple case based 
on a density argument which is standard. However, we couldn't find a reference for this, so we include a proof in 
\cref{app:density} (see \cref{lem:gen}). 

\begin{theorem}[\cite{Kuranishi51, Bois09}+ \cref{lem:gen}]\label{thm:gen2}
Let $\fkg$ be a semisimple Lie algebra $\fkg$ over 
an algebraically closed %
field $\F$ of 
characteristic not $2$ or $3$.  Then $\fkg$ can be generated by two elements.
\end{theorem}

\cref{thm:lin_main} follows from \cref{thm:linear_kernel} and \cref{thm:gen2}. \qed

\paragraph{Acknowledgement.} {T. M. would like to thank Pallav Goyal for helping him 
understand Lie algebras. The authors would like to thank Visu Makam for 
communicating \cite{DM21} to them, and the anonymous reviewers for their detailed 
and helpful feedback.}

\appendix

\section{Basic notions for Lie algebras and its representations}\label{app:basic}

A \emph{Lie algebra} is vector space $\fkg$ with an alternating bilinear map, 
called a 
Lie bracket, 
$[-,-]:\fkg\times \fkg\to\fkg$ that satisfies the Jacobi identity 
$\forall x, y, z\in \fkg, 
[x,[y,z]]+[z,[x,y]]+[y,[z,x]]=0.
$
A subalgebra of a Lie algebra $\fkg$ is a vector subspace which is closed under 
the Lie bracket. 

Given two Lie algebras $\fkg$ and $\fkh$, a \emph{Lie algebra homomorphism} is a 
linear 
map respecting the Lie bracket, i.e., a linear map $\phi:\fkg\to\fkh$ such that 
$\phi([a, b])=[\phi(a), \phi(b)]$.

Given a vector space $V$, we use $\fkgl(V)$ to denote the Lie algebra, which 
consists of linear endomorphisms of $V$ with the Lie bracket $[A, B]=AB-BA$ for 
$A, B\in \fkgl(V)$. 

A \emph{representation of a Lie algebra} $\fkg$ is a Lie algebra homomorphism 
$\phi: \fkg\to \fkgl(V)$ for some vector space $V$. A subspace $U\leq V$ is 
\emph{invariant} under $\phi$, if for any $a\in \fkg$, $\phi(a)(U)=U$. We say that 
$\phi$ is \emph{irreducible}, if the only invariant subspaces under 
$\phi$ are the zero space and the full space. We say that $\phi$ is 
\emph{completely reducible}, if there exists a proper direct sum decomposition 
of 
$V=V_1\oplus \dots \oplus V_c$, such that each $V_i$ is invariant under $\phi$. 

A representation $\phi:\fkg\to\fkgl(V)$ is 
trivial, if $\phi(x)=0\in \fkgl(V)$ for any $x\in \fkg$. In this case, when $V$ is 
of dimension $1$, $\phi$ is the \emph{trivial irreducible representation}. The 
adjoint 
representation of $\fkg$, $\ad:\fkg \to \fkgl(\fkg)$ 
is defined as $\ad_x(y) = [x,y]$ for $x, y\in \fkg$.  

Suppose $V$ is of dimension 
$n$ over a 
field $\F$. After fixing a basis of $V$, $\fkgl(V)$ can be identified as $\M(n, 
\F)$. Then the image of a Lie algebra representation $\phi$ is a matrix subspace 
$\cB\leq 
\M(n, \F)$ that is closed under the natural Lie bracket $[A, B]=AB-BA$ for $A, 
B\in \cB$.

\section{Correspondences between Lie algebras and Lie 
groups}\label{app:corresp}

Lie algebras are closely related to Lie groups. In the case of finite dimensional 
complex and real Lie algebras, there is a tight correspondence. Since matrix 
Lie algebras are the main object of study in this article, we only need results 
for matrix Lie algebras 
and matrix Lie groups, and not the most general definitions. In the following, we 
present some basic facts about the correspondence between Lie algebras and Lie 
groups 
in the matrix setting.

We follow \cite{hall} for the definitions and some basic results about matrix Lie 
groups and Lie algebras over $\C$ that we will use later.

A \emph{matrix Lie group} is a subgroup $G$ of 
$\GL(n, \C)$  with the property that if $(A_m)_{m\in\N}$ is any sequence of 
matrices in $G$, and $A_m$ converges to some matrix $A$, then either $A$ is in $G$ 
or $A$ is non-invertible.

For $X \in \M(n, \C)$, define the exponential by the usual power series, that is,  
$e^X = \sum_{i=0}^\infty \frac{X^i}{i!}$. By \cite[Proposition 2.1]{hall}, this 
power 
series converges absolutely for any $X\in\M(n, \C)$, and $e^X$ is a 
continuous function of $X$.
A straightforward consequence of the absolute convergence is that we can 
differentiate term by term, 
which implies that $\frac{d}{dt}e^{tX} = Xe^{tX} = e^{tX}X$.

Given a matrix Lie group $G$, the associated Lie algebra $\Lie(G)$ is defined as 
$\Lie(G) = \{X\in \M(n,\C)\;|\; \forall t \in \R, e^{tX} \in G\}$. Let 
$\fkg$ denote 
$\Lie(G)$; this notation is consistent with our previous notation. Clearly, for 
any 
$M\in \fkg$, the one-parameter group $\{e^{tM} \;|\; t\in \R\}$
is a 
subgroup of 
$G$.

We need the following two classical results relating matrix Lie groups and matrix 
Lie algebras in \cref{sec:shrunk}.

\begin{theorem}[{\cite[Theorem 5.20]{hall}}]\label{thm:subalgebra}
 Let $G$ be a matrix Lie group with Lie algebra $\fkg$ and let $\fkh$ be a Lie 
 subalgebra of $\fkg$. Then there exists a unique connected Lie subgroup $H$ of 
 $G$ with Lie algebra $\fkh$. In particular, every matrix Lie algebra $\fkg$ is 
 the Lie algebra of a Lie group. 
\end{theorem}

\begin{theorem}[{\cite[Theorem 3.20 (1)]{hall}}]\label{expstab}
Let $G$ be a matrix Lie group, and let $\fkg = \Lie(G)$. Then for any $X \in \fkg$ 
and $g 
\in G$, we have $gXg^{-1} \in \fkg$. 
\end{theorem}

\section{Some results about Cartan subalgebras}\label{app:cartan}

\paragraph{Cartan subalgebras.} 
Let $\fkg$ be a Lie algebra. A \textit{subalgebra} 
$\fkh \subseteq \fkg$ is a vector subspace that is closed under the Lie bracket 
(inherited from $\fkg$). In other words,  $[\fkh, \fkh] \subseteq \fkh$.  An 
\textit{ideal} $\mathfrak{i} \subseteq \fkg$ is a subalgebra such that $[\fkg, 
\mathfrak{i}] \subseteq \mathfrak{i}$.  Let $\fkg_1$ and $\fkg_2$ be ideals of 
$\fkg$. Define $[\fkg_1, \fkg_2] = 
\mathrm{span} \left( [x,y] \;|\; x\in\fkg_1, y \in \fkg_2 \right)$. Let $\fkg^1 = \fkg$ and inductively 
define $\fkg^i = [\fkg^{i-1}, \fkg]$. An algebra $\fkg$ is called 
\emph{nilpotent} if there is an $n$ such that $\fkg^n = 0$.  Similarly, define 
$\fkg^{(1)} = \fkg$ and $\fkg^{(i)} = [\fkg^{(i-1)}, \fkg^{(i-1)}]$.  An algebra 
$\fkg$ is called \emph{solvable} if there is an $n$ such that $\fkg^{(n)} = 0$. 
The \emph{normalizer} of a subspace $\mathfrak{a}$ of $\fkg$  is defined as 
$\mathfrak{n}_\fkg(\mathfrak{a}) = \{x \in \fkg \;|\; [x,\mathfrak{a}] \subseteq 
\mathfrak{a} \}$. A subalgebra $\fkh$ of $\fkg$ is a \emph{Cartan subalgebra} if 
it is nilpotent and $\mathfrak{n}_\fkg(\fkh)  =  \fkh$. 

We shall need the following classical result on Cartan subalgebras. For $x\in 
\fkg$, recall that 
$\ad_x:\fkg \to 
\fkg$ is the linear map defined by  
$\ad_x(y) = [x,y]$ for $y\in \fkg$. In particular, the exponentiation $e^{\ad_x}$ 
is a linear map from $\fkg$ to $\fkg$, and it is a Lie algebra automorphism if 
$\ad_x$ is nilpotent, called an inner automorphism. The group generated by inner 
automorphisms is denoted by $\mathrm{Int}(\fkg)$.

\begin{theorem}[{See e.g. \cite[Chapter 3.5]{Gra00}}]\label{conj} Let $\fkg$ be a 
Lie algebra over an algebraically closed field $\F$ of characteristic zero. For 
any two Cartan subalgebras 
$\fkh_1$ and $\fkh_2$, there exists $g\in\mathrm{Int}(\fkg)$ such that
$\fkh_1 = g(\fkh_2) $.

\end{theorem}
To recover the statement in \cref{lem:key}, note that 
for a matrix Lie 
algebra $\cB\leq\M(n, \C)$, an inner automorphism takes the form as a conjugation 
by an invertible matrix. 
This is because $\mathrm{Ad}(e^x)=e^{\ad_x}$, where
$\mathrm{Ad}$ is the conjugation by matrices. This can be seen, e.g.,
by taking the derivative of
$\mathrm{Ad}(e^x)Y=e^{tx}Ye^{-tx}$ at $t=0$.

\begin{theorem}[Lie's theorem on solvable Lie algebras]\label{thm:Lie} 
Let $\F$ be an algebraically closed field of characteristic zero.
Let $\cB\leq\M(n, \F)$ be a solvable matrix Lie algebra over $\F$. Then there 
exists 
$T\in\GL(n, \F)$, such that for any $B\in \cB$, $TBT^{-1}$ is upper triangular.
\end{theorem}

\paragraph{Regular elements of Lie algebras.} 
Let $\lambda$ be a formal variable, and let
$\sum_i c_{i, x}\lambda^i$ be the characteristic polynomial of 
$\ad_x$.
The smallest $r$ 
such that $c_{r, x}$ is not identically zero over all $x\in \fkg$ is called the 
\textit{rank} of $\fkg$. The 
open set of points $\{ x \in \fkg  \;|\; c_r(x) \neq 0\}$ is the set of 
\textit{regular points}. A simple observation is that the set of regular elements 
is Zariski open and thus it is dense.  For $x \in  \fkg$, the Fitting null component of $\ad_x$ is \[\FitZero(\ad_x) 
=\{y\in \fkg : \exists m>0, \ad_x^m(y) = 
0\}.\]

Regular elements and Cartan subalgebras are closely related as the following 
theorem shows. 

\begin{theorem}[{\cite[Corollary 3.2.8]{Gra00}}]\label{reg} Let $\fkg$ be a Lie 
algebra over a field of order larger than $\dim(\fkg)$. For a regular 
$x\in\fkg$, 
$\FitZero(\ad_x)$ is a Cartan subalgebra.
\end{theorem}

\paragraph{Computing Cartan subalgebras.} We shall need the following result of de 
Graaf, Ivanyos, and R\'onyai \cite{cartan} regarding computing Cartan subalgebras. 
In algorithms, Lie algebras are often given by structure constants. That is, 
let $\fkg$ be a Lie algebra of dimension $n$ over a field $\F$, and 
let $a_1, 
\dots, a_n$ be a linear basis of $\fkg$. The \emph{structure constants} 
$\alpha_{ijk}$ 
($i,j,k\in \{1,\ldots,n\}$)
are field elements such that $[a_i,a_j]=\sum_{k=1}^n\alpha_{ijk}a_k$.

\begin{theorem}[{\cite[Theorem 5.8]{cartan}}]\label{thm:algoregular}
Let $\fkg$ be a Lie algebra of dimension $n$ over a field $\F$ with $|\F|>n$. 
Suppose $\fkg$ is given by its structure constants with respect to a basis $a_1, 
\dots, a_n$, and fix $\Omega\subseteq \F$ such that $|\Omega|=n+1$. Then there is 
a deterministic polynomial-time algorithm 
which
computes a 
regular element $x = \sum \alpha_ia_i,\; \alpha_i\in \Omega$, such that 
$\FitZero(\ad_x)$
is a Cartan subalgebra of $\fkg$.
\end{theorem}

Note that to obtain \cref{thm:cartan}, we start with a matrix Lie algebra 
$\cB=\langle B_1, \dots, B_m\rangle\leq\M(n, \F)$, compute structure constants by 
expanding $[B_i, B_j]=\sum_{k\in[m]}\alpha_{i,j,k}B_k$, apply 
\cref{thm:algoregular}, and use its output to obtain a subspace of $\cB$ 
which is a Cartan subalgebra.

\section{Density arguments and the generation of Lie algebras}\label{app:density}

In this part we prove some facts based on standard density arguments.

Let $U$ be an $m$-dimensional vector space over an infinite field $\F$.
By choosing a basis we identify $U$ with $\F^m$. We say that a nonempty
subset $D$ of $U$ is {\em huge} if there exists a nonzero polynomial 
$f(t_1,\ldots,t_m)\in \F[t_1,\ldots,t_m]$ such that if 
$u=(u_1,\ldots,u_m)^T\not\in D$ then $f(u_1,\ldots,u_m)=0$. 
(Thus, huge subsets are those that contain Zariski open
subsets.)
It is easy to see that hugeness is independent of the choice of
the basis and that the intersection of finitely many huge subsets
is huge as well. As a hyperplane of $U$ consists of the zeros
of a linear function on $U$, we have that any huge subset of 
$U$ spans $U$.

Let $\fkg$ be an $m$-dimensional Lie algebra over $\F$. 
Let $u_1,\ldots,u_m$ be a basis for $\fkg$. Recall that
the structure constants $\alpha_{ijk}$ ($i,j,k\in \{1,\ldots,m\}$)
are field elements such that such that $[u_i,u_j]=\sum_{k=1}^m\alpha_{ijk}u_k$.
A {\em Lie expression}
or Lie polynomial $E(z_1,\ldots,z_\ell)$ in $\ell$ variables $z_1,\ldots,z_\ell$
is an expression that can be recursively built using linear
combinations and the bracket symbol. 
Let $x_i=\sum_{j=1}^mx_{ij}u_j$
($i=1,\ldots,\ell$). Then the 
structure constants %
 can be used to expand
$E(x_1,\ldots,x_\ell)$ as a vector whose coordinates are
polynomials in $x_{ij}$. If we assign $m-1$ elements of $\fkg$
to the variables $z_2,\ldots,z_\ell$ then $E$ expands to
an a vector whose coordinates are polynomials in $x_{11},\ldots,x_{1\ell}$
that may include nonzero constant terms. Therefore it will be
convenient to also consider
Lie expressions {\em over $\fkg$}: these are expressions which 
may include constant elements from $\fkg$.  
{}From the definition of density it follows that if $E$ is an expression in a single 
variable
$z$ that is not identically zero on $\fkg$ then the elements $x$ of
$\fkg$ on which $E$ evaluates to a nonzero element of $\fkg$ is huge.
Furthermore, if there are $m$ expressions $E_1(z),\ldots,E_m(z)$
such that there exists an element $x\in \fkg$ such that
$E_1(x),\ldots,E_m(x)$ are linearly independent then 
such elements are a huge subset of $\fkg$. To see this,
just consider the determinant expressing that $E_1(x)\ldots,E_m(x)$
are linearly dependent.

\begin{proposition}\label{lem:span2}
Let $\fkg$ be an $m$-dimensional Lie algebra over
a large enough field $\F$ such that there exist two
elements $x_0,y_0\in\F$ that
generate $\fkg$. Then there are elements $x_i,y_i\in \fkg$,
($i=1,\ldots,m^2$) such that $x_i$ and $y_i$ generate
$\fkg$ for every $i$ and that $x_i\otimes y_i$ span $\fkg\otimes \fkg$.
\end{proposition}

\begin{proof}
Pick expressions $E_i(z,w)$ $(i=1,\ldots,m)$
such that $E_i(x_0,y_0)$ are linearly independent. Then the set
of elements $x$ such that $E_i(x,y_0)$ are linearly independent
is huge and hence contains a basis $x_1,\ldots,x_m$ of $\fkg$.
For each $j$, the subset consisting of those $y$ for which
$E_i(x_j,y)$ are linearly independent is a huge set $D_j\subset \fkg$
whence there exist elements $y_{jk}\in D_j$ ($k=1,\ldots,m$) that
are a basis for $\fkg$. Each of the $m^2$ pairs $x_j,y_{jk}$ generate $\fkg$.
To see that they span $\fkg\otimes \fkg$, write an element $z\in \fkg\otimes \fkg$
in the form $z=\sum x_j\otimes y_j'$ and express $y_j'$ as
$y_j'=\sum_k \alpha_{jk}y_{jk}$. Then $z=\sum_{j,k} \alpha_{jk}x_j\otimes
y_{jk}$. 
\end{proof}

\begin{lemma}\label{lem:gen}
Let $\fkg_1,\ldots,\fkg_m$ be finite dimensional simple Lie
algebras over a large enough field $\F$, each generated by
2 elements. Then $\fkg_1\oplus\ldots\oplus \fkg_m$ is also
generated by two elements. 
\end{lemma}

\begin{proof}
Assume that $x_i$ and $y_i$ generate $\fkg_i$ ($i=1,\ldots,m$).
We claim that we may further assume that 
$\ad(x)^{d_i}y_i\neq 0$ where $d_i=\dim_\F(\fkg_i)$.
Indeed, if $\ad(x)^{d_i}y_i= 0$ then, by Engel's theorem,
 there exists a pair $(w_i,z_i)\in \fkg\times \fkg$ such that 
$\ad(w_i)^{d_i}z_i\neq 0$. 
If we fix $z_i$ from such a pair then the elements $w_i$ such that 
$\ad(w_i)^{d_i}z_i\neq 0$ is a huge set. There
exist $d_i$ Lie expressions $E_1,\ldots,E_{d_i}$ in in two variables
such that $E_1(x_i,y_i),\ldots,E_{d_i}(x_i,y_i)$ are linearly independent
elements of $\fkg_i$. The elements $w_i$ such that
$E_1(w_i,y_i),\ldots,E_{d_i}(x_i,y_i)$ are linearly independent
form a huge set. The intersection of these two huge sets
is still huge and hence non-empty. We replace $x_i$ with
an element from the intersection. Now the set of $w_i$ such that
$\ad(x_i)^d_iw_i\neq 0$ is huge as well as these set of
those for which $E_j(x_i,w_i)$ are linearly independent. We can 
replace $y_i$ with an element from the intersection.

Let $f_i=f_i(t)$ be the monic polynomial of smallest degree such that 
$f_i(\ad(x_i))y_i=0$. Note that $f_i$ has degree at most $d_i$
and the assumption on $x_i$ and $y_i$ implies that $f_i$ is not a divisor
of $t^{d_i}$. Therefore each $f_i$ has a nonzero root
(in the algebraic closure $\overline \F$ of $\F$).
Let $R_i$
be the set of nonzero roots of $f_i$ in $\overline \F$. 
There exist field elements
$\alpha_1,\ldots,\alpha_m\in \F$ such that 
the sets $\alpha_iR_i$ are pairwise disjoint. Replacing
$x_i$ with $\alpha_ix_i$ we arrange that the sets 
$R_i$ become pairwise disjoint.
Then for each $i$, put $h_i=\prod_{j\neq i}f_j$. We have
that $h_i(\ad(x_j))y_j=0$ for every $j\neq i$, while $h_i(\ad(x_j))y_i\neq 0$ as $h_i$ is not divisible by $f_i$.

Put $x=\sum_{i=1}^mx_i$ and $y=\sum_{i=1}^my_i$. Then
$h_i(x)y$ is a nonzero element of $\fkg_i$. Let $M$
be the subalgebra of $\fkg$ generated by $x$ and $y$.
We see that $M$ has a nonzero element, say $z_i$ contained in
$\fkg_i$. The projection of $M$ on the $i$th component is
clearly $\fkg_i$ and $x_i$ and $y_i$ generate $\fkg_i$. As
$\fkg_i$ is simple we have that the ideal of $M$ generated by
$z_i$ is $\fkg_i$. This holds for all $i=1,\ldots,m$,
showing that $M=\fkg$.
\end{proof}

\section{Linear kernel vectors of matrix Lie algebras}\label{app:old}
In this section,  we will give an alternative proof of 
\cref{thm:linear_kernel}, which doesn't use density arguments, but instead 
uses weight decomposition of representations of semi-simple Lie algebras over $\C$.

\subsection{Weight decomposition of Lie algebra representations}
Fix a Cartan subalgebra $\fkh$ of a semisimple Lie algebra $\fkg$ over $\C$. 
By definition $\fkh$ is nilpotent. 
If $\fkg$ is semisimple, $\fkh$ is abelian \cite[Thm3, Ch.3]{serre}.  Similar to 
the notion of eigenvalues and eigenspaces is the concept of a weight and its 
weight space.  Intuitively,  it can be thought of as a linear function that 
captures the eigenvalues of a set of matrices simultaneously.   
Formally,  a \emph{weight} is an element of 
$\fkh^*$. If $w \in \fkh^*$, then the $w$-weight space of $V$ is defined as $V_{w} 
= \{v\in V\;|\; \forall \; h \in \fkh, \rho(h)\cdot v = w(h)v\; \}$.

\begin{theorem}
If $\fkg$ is a complex semi-simple Lie algebra, then every representation 
$(\rho,V)$ can be decomposed into weight spaces $V = \oplus_w V_{w}$. 
\end{theorem}

Using this decomposition we have a basis such that for any $h \in \fkh$, $\rho(h)$ 
is a diagonal matrix with $w(h)$  as diagonal elements where $w$ runs over all 
weights of $V$.

\begin{fact}\label{fact:singular}
The matrix space defined by the image $\rho(\fkg)$ of the representation 
$(\rho,V)$ is singular iff $0$ is a weight of the representation, i.e. $V_0$ 
occurs with multiplicity at least one.
\end{fact}
\begin{proof}
This follows easily from the observation about $\rho(h_i)$ which implies that 
$\rho(\fkh)$ is singular if $0$ is a weight. From \cref{sec:maxrk}, we know 
that the entire algebra is singular if any of its Cartan subalgebra is.  If $0$ is 
not a weight, then it is easy to construct a element $h \in \fkh$ such that $w(h) 
\neq 0$ for every weight. Thus, $\rho(h)$ has full rank.  
\end{proof}

If we decompose the adjoint representation,
the weights we obtain are called \emph{roots} usually denoted by $\Phi$. It is 
also a fact that if $\alpha \in \Phi$, then $-\alpha \in \Phi$. We thus can write 
$\fkg = \fkh\oplus_{\alpha\in \Phi}\fkg_\alpha $. Moreover, each of the spaces 
$\fkg_\alpha$ is one-dimensional. We denote an element of $\fkg_\alpha$ by 
$g_\alpha$ which is unique upto a scalar.  Such a decomposition of the Lie algebra 
is very useful as we can understand the action under any representation of these 
subspaces $g_\alpha$ as follows . 

\begin{proposition}\cite[Prop.1, Chapter 7]{serre} \label{prop:weightaction}
For any representation $(\rho,V)$ of $\fkg$,  $\rho(g_\alpha)V_w \subset 
V_{w+\alpha} $ for every weight $w$ and every root $\alpha$.
\end{proposition}

\subsection{ Notation }
Fix a complex semi-simple Lie algebra $\fkg$ and a Cartan subalgebra $\fkh$.  Let 
$\Phi$ be its set of roots.  We choose a Cartan-Weyl basis\footnote{Cartan-Weyl 
basis and the Chevally basis differ only by a normalization. We do not need 
properties of the coefficients $c_{\alpha,\beta}$ and thus either basis works 
well.} for $\fkg$ (cf. e.g. 
\cite[pp. 48]{serre}). This means that we have a set of simple roots\footnote{Simple roots are just a basis of $\fkh^*$ while the set of all roots can be linearly dependent. } $S = 
\{\alpha_1, \cdots \alpha_n\}$ and a basis of $\fkh$, 
$\{h_1,\cdots h_n\}$ such that the following hold,
\begin{align*}
    [h_{i},g_{\alpha_j}] &= \alpha_j(h_i)g_{\alpha_j} &\forall i,j \in [n]\\
    [g_{\alpha},g_{\beta}] &= c_{\alpha\beta}g_{\alpha +\beta} \;\; ( 
    c_{\alpha\beta}\neq 0) & \alpha + \beta \in \Phi \\
    [g_{\alpha},g_{\beta}] &= 0 & \text{if } \alpha + \beta \not \in \Phi\\
    [g_{\alpha_i},g_{-\alpha_i} ] &= h_i &\forall i \in [n]
\end{align*}

Choose a basis of $V \cong \F^N$, such that $\rho(\fkh)$ is diagonal. Let $W$ be 
the set of weights of $V$ and thus $V = \oplus_{w \in W} V_{w}$ such that $V_w$ is the $w$ weight space of $\fkh$. Note that we are assuming that $0\in W$ as 
non-singular spaces anyway cannot have a linear kernel.  

\subsection{Main proof}

We recall that given a Lie algebra $\fkg$ and a representation $(\rho,V)$ a linear kernel vector $\phi:\fkg \to V $ is a linear map such that $\rho(x)\psi(x) = 0$ for every $x \in \fkg$. We state the main lemma we need and will prove it later. 

\begin{lemma}\label{lem:action} 
Assume that trivial representation is not a subrepresentation of 
the representation $(\rho,V)$ of $\fkg$.  Let $\psi:\fkg \to V$ be a linear kernel vector. 
Then for any $\alpha,\beta \in \Phi$ such that $\alpha+\beta \neq 0$ and $h \in \fkh$ we have 
\begin{align*}
\psi(h) \in V_0, \; &\psi(g_\alpha) \in V_\alpha \text{  and}\\
\psi([g_\alpha,g_\beta])&=\rho(g_\alpha)\psi(g_\beta) \\
\psi([h,g_\alpha]) & = \rho(h)\psi(g_\alpha)
\end{align*}
\end{lemma}

\begin{theorem}
Assume that trivial
representation is not a subrepresentation of 
the representation $(\rho,V)$ of $\fkg$. Then for every $x,y\in \fkg$,
\begin{equation*}\label{eqn:kern_deficit}
\psi([x,y])-\rho(x)\psi(y)=0
\end{equation*}
\end{theorem}

\begin{proof}
By linearity, it suffices to show this for the basis vectors $h_i, g_\alpha$.  \cref{lem:action} shows it in every except when we have $(x,y) = (g_\alpha, g_{-\alpha})$. Fix any root $\alpha$. By \cref{lem:action} and \cref{prop:weightaction},  the vector $\psi([g_\alpha, g_{-\alpha}]) - \rho(g_\alpha)\psi(g_{-\alpha}) \in V_0$ and thus is annihilated by $\rho(\fkh)$. We now wish to show that it is annihilated by $\rho(g_\beta)$ for any root $\beta$.  The assumption that there are no trivial subrepresentations then implies that this vector must be zero.

\cref{lem:linker-key} already shows it if $\beta\in \{\alpha, -\alpha \}$ and so we assume that's not the case. 
\begin{align*}
 \rho(g_\beta)\rho(g_\alpha)\psi(g_{-\alpha}) &= \rho([g_\beta,g_\alpha])\psi(g_{-\alpha}) + \rho(g_\alpha)\rho(g_\beta)\psi(g_{-\alpha}) & \rho  \text{ is a representation  } \\
 &= \rho([g_\beta,g_\alpha])\psi(g_{-\alpha}) + \rho(g_\alpha)\psi([g_\beta,g_{-\alpha}]) &\text{ \cref{lem:action} for } (\beta, -\alpha)   \\
 &= \psi( [[g_\beta,g_\alpha] ,g_{-\alpha}] ) + \psi([g_{\alpha},  [g_\beta, g_{-\alpha}] ] ) &\text{ \ref{lem:action} for } (\alpha +\beta, -\alpha),  (\alpha,  \beta - \alpha)  \\
 &= \psi\left( [g_\beta,g_\alpha] ,g_{-\alpha}] +  [g_{\alpha},  [g_\beta, g_{-\alpha}] ]  \right) &\text{ By linearity}\\
  &=\psi\left( [g_\beta ,[g_\alpha ,g_{-\alpha}]] \right)  &\text{By Jacobi identity} \\
  &= \rho(g_\beta)  \psi\left( [g_\alpha ,g_{-\alpha}] \right) &\text{ \ref{lem:action} for } (\beta,  h),\;  h = [g_\alpha,g_{-\alpha}] \in \fkh
\end{align*}
Here, we have used \cref{lem:action} formally even if one of them is not a root to prevent dividing into cases. For example,  if $\beta -\alpha$ is not a root then that term is anyway $0$ and we can represent $0$ as $\psi(0)= \psi([g_\beta, g_{-\alpha}])$. 
\end{proof}

\subsection{Proof of \cref{lem:action} }

For ease of notation,  we label $H_i = \rho(h_i)$ 
for $1 \leq i \leq n$ and $X_\alpha = \rho(g_\alpha),\; \alpha \in \Phi$.  Similarly, we will have $v_i := \psi(h_i),\; v_\alpha := \psi(g_\alpha)$.  We restate \cref{eq:linker-cross} in more verbose terms,
\begin{align}
\label{eq:linkernelh}H_iv_i&= 0  &\forall i \in [n]\\
\label{eq:linkernelg}X_\alpha v_\alpha &= 0 &\forall \alpha \in \Phi\\
\label{eq:linkernelhh}H_iv_j + H_jv_i &= 0  &\forall i,j \in [n], i\neq j \\
\label{eq:linkernelhg}H_iv_\alpha + X_\alpha v_i &= 0  &\forall i \in [n], \alpha 
\in \Phi \\
\label{eq:linkernelgg}X_\beta v_\alpha + X_\alpha v_\beta& = 0  &\forall 
\alpha,\beta \in \Phi
\end{align}

\begin{lemma}[Structure]\label{lem:viin0}
For every $i \in [n]$,  $v_i  \in V_0$, and for every root $\alpha$, $ v_\alpha \in V_\alpha$ if $\alpha$ is  a weight and is $0$ otherwise.
\end{lemma}
\begin{proof}
i) Let $v_j = \sum_{w \in W} u_w$ where $ u_w \in V_w$.
Since $0=H_j v_j = \sum_{w \in W} w(h_j)u_w $, we have that $w(h_j)u_w = 0$. For any $w\neq 0$ such that $u_w \neq 0$ we have $w(h_j) = 0$. Pick $k \in [n]$ such that $w(h_k) \neq 0$. Then, $   H_kv_j +H_jv_k = 0$.  Looking at the $V_w$ component $w(h_k)u_w + w(h_j)v_k = 0$. Since, $w(h_j) = 0$, we get that $w(h_k)u_w = 0$. But $k$ is chosen such that $w(h_k) \neq 0$ and thus, $u_w = 0$.

ii) Fix an $\alpha$. We just proved that $\forall i,\; v_i  \in V_0$ and using 
\cref{prop:weightaction} we get $ X_\alpha v_i
\in V_\alpha$.  Suppose $v_\alpha = \sum_{w \in W} u_w$, where $u_w \in V_w$. For any non-zero $w \neq \alpha$, pick $i$ such that $w(h_i) \neq 0$. Then,  $H_iv_\alpha + 
X_\alpha v_i= 0$. But we already know that $X_\alpha v_i \in V_\alpha$ and thus $H_i v_\alpha \in V_\alpha$.  Suppose $H_i v_\alpha = \sum_{w \in W} w(h_i)u_w \in V_\alpha$. 

Then, the $V_w$ component should be zero but $w(h_i) \neq 0 \implies u_w = 0$. It follows that $v_\alpha \in V_\alpha \oplus V_0$ for every root $\alpha$.  Fix $\alpha$ and now for every $\beta \neq \alpha$,  we have $X_\alpha v_\beta + X_\beta v_\alpha =  0 $. Comparing the $V_\beta$  component we get that $X_\beta u_0 = 0$. This is true for every $\beta$ and since $u_0 \in V_0$, it is also true for $\rho(h)$ for every $h \in \fkh$. Thus,  for every $x \in \fkg$,  $\rho(x)u_w$ and since there are no trivial submodules,  $u_0 = 0$. Thus,  $v_\alpha \in V_\alpha$
\end{proof}

\begin{lemma}\label{lem:sum}
For all pairs of roots $\alpha,\beta$ such that  $ 0 \neq \beta + \alpha \in 
\Phi$,  $X_\beta v_\alpha  = c_{\alpha\beta} v_{\alpha +\beta}$.
\end{lemma}
\begin{proof}
We have that $H_iv_{\alpha+\beta} + X_{\alpha+\beta}v_i = 0$.
Similarly, $X_\alpha v_i = -H_iv_{\alpha} = -\alpha(h_i)v_\alpha $ and $X_\beta 
v_i = -H_iv_{\beta} = -\beta(h_i)v_\beta $. Moreover, $X_\beta v_\alpha + X_\alpha 
v_\beta = 0$.

Now, $X_{\alpha+\beta} = \rho(g_{\alpha+ \beta}) = c\rho([g_\alpha,g_{\beta}]) =  
c(X_{\alpha}X_{\beta} - X_{\beta}X_{\alpha})$ where $c = 
\frac{1}{c_{\alpha\beta}}$. Thus,
\begin{align*}
-(\alpha+\beta)(h_i)v_{\alpha+\beta} &= X_{\alpha+\beta}v_i\\
&= c(X_{\alpha}X_{\beta} - X_{\beta}X_{\alpha}) v_i\\
&= c(X_{\alpha}(X_{\beta}v_i) - X_{\beta}(X_{\alpha}v_i))\\
&= c(X_{\alpha}(-\beta(h_i)v_\beta) - X_{\beta}(-\alpha(h_i)v_\alpha)\\
&= -c\left( \beta(h_i) X_{\alpha}v_\beta - \alpha(h_i) X_{\beta}v_\alpha \right)\\
&= -c\left( -\beta(h_i) X_{\beta}v_\alpha - \alpha(h_i) X_{\beta}v_\alpha \right)\\
&= c\cdot (\alpha+\beta)(h_i)\cdot X_{\beta}v_\alpha.
\end{align*}
Thus, picking $h_i$ such that $(\alpha+\beta)(h_i) \neq 0$, we get that 
$X_{\beta}v_\alpha = \frac{1}{c}v_{\alpha+\beta} $.
\end{proof}
\begin{lemma}\label{lem:alphabeta}
Let $\alpha,\beta$ be roots such that $\beta \neq -\alpha,\alpha+\beta \not \in 
\Phi $ Then. we have that $X_\alpha v_\beta = 0$. 
\end{lemma}
\begin{proof}
Since, $\alpha+\beta \not\in \Phi,\; [g_{\alpha},g_{\beta}] = 0$ which implies 
that $\rho([g_{\alpha},g_{\beta}]) = 0 $ and thus, $X_\alpha X_{\beta} = 
X_{\beta}X_{\alpha}$. Moreover, $\exists h_i $ such that $\alpha(h_i) \neq 
-\beta(h_i)$ because the $h_i$ form a basis for $\fkh$. We fix our $i$ to be one 
such. Now, from \cref{eq:linkernelgg}, we get that 
$X_{\alpha}v_\beta+ X_\beta v_\alpha = 0  $. 
From \cref{eq:linkernelhg}, we get that $H_iv_\alpha + X_\alpha v_i = 0 $ 
and multiplying it by $X_{\beta}$ we obtain, $\alpha(h_i)X_{\beta}v_\alpha + 
X_{\beta}X_\alpha v_i = 0 $. Repeating it with $\beta$ and $\alpha$ switched, we 
get, $\beta(h_i)X_{\alpha}v_\beta+ X_{\alpha}X_\beta v_i = 0 $. Subtracting these 
2 equations, we get $\alpha(h_i)X_{\beta}v_\alpha - \beta(h_i)X_{\alpha}v_\beta = 
0 $. We already have another equation i.e. $X_{\alpha}v_\beta+ X_\beta v_\alpha = 
0$.
Since $\beta(h_i) \neq -\alpha(h_i)$, these two homogeneous equations are 
independent and thus, the only solution is that $X_{\alpha}v_{\beta} = 
X_{\beta}v_{\alpha} = 0  $.  
\end{proof}

The structure lemma  establishes \cref{lem:action} when $y \in \fkh$ i.e.  for any $x \in \fkg, h \in \fkh$ we have $\rho(h)\psi(x) = \psi([h,x])$. To see this notice that $\rho(h)\psi(g_\alpha) = \alpha(h)\psi(g_\alpha)  = \psi( \alpha(h)g_\alpha)  )  = \psi([h,g_\alpha])$ where the first equality uses that $\psi(g_\alpha) \in V_\alpha$ and the last by the property of the basis.  And the other two lemmas extend it to $\alpha,\beta$ as $c_{\alpha\beta}v_{\alpha+\beta} = c_{\alpha\beta}\psi( g_{\alpha+\beta}) = \psi( [g_{\alpha},g_{\beta}])  $.

\bibliographystyle{alpha}
\bibliography{references}

\end{document}

%% file: preamble.tex
\usepackage[left=1.25in,right=1.25in,top=1in,bottom=1in]{geometry}
\usepackage{latexsym}
\usepackage{amsmath}
\usepackage{amssymb}
\usepackage{amsthm}
\usepackage{hyperref}
\usepackage{cite}
\usepackage{graphicx}
\usepackage{xcolor}
\usepackage{enumitem}
\usepackage{braket}
\usepackage{bm}
\usepackage{xspace}
\setlist[description]{font=\normalfont\itshape\textbullet\space}
\usepackage[all]{xy}

\renewcommand{\paragraph}[1]{\vspace{6pt} \noindent \textbf{#1}\xspace}

\usepackage[T1]{fontenc}
\usepackage{mathtools,dsfont,bbm}
\usepackage{mathpazo}
\theoremstyle{plain}
\newtheorem{theorem}{Theorem}[section]

\newtheorem{corollary}[theorem]{Corollary}
\newtheorem{lemma}[theorem]{Lemma}

\newtheorem{proposition}[theorem]{Proposition}

\newtheorem{fact}[theorem]{Fact}

  {\end{tabular}\par\medskip}

\theoremstyle{definition}

\newtheorem{remark}[theorem]{Remark}

\newtheorem{example}[theorem]{Example}

\newcommand{\GL}{\mathrm{GL}}
\newcommand{\F}{\mathbb{F}}

\newcommand{\Q}{\mathbb{Q}}
\newcommand{\C}{\mathbb{C}}
\newcommand{\R}{\mathbb{R}}
\newcommand{\N}{\mathbb{N}}

\newcommand{\crk}{\mathrm{mrk}}

\newcommand{\im}{\mathrm{im}}

\newcommand{\ad}{\mathrm{ad}}

\newcommand{\poly}{\mathrm{poly}}

\newcommand{\ncrk}{\mathrm{ncrk}}

\newcommand{\M}{\mathrm{M}}

\newcommand{\Stab}{\mathrm{Stab}}

\newcommand{\NP}{\mathrm{NP}}
\newcommand{\coNP}{\mathrm{coNP}}

\newcommand{\sd}{\mathrm{sd}}

\newcommand{\spa}[1]{\mathcal{#1}}

\newcommand{\cA}{\spa{A}}
\newcommand{\cB}{\spa{B}}
\newcommand{\cC}{\spa{C}}

\newcommand{\cH}{\spa{H}}

\newcommand{\fkg}{\mathfrak{g}}
\newcommand{\fkh}{\mathfrak{h}}
\newcommand{\fkgl}{\mathfrak{gl}}
\newcommand{\fksl}{\mathfrak{sl}}

\newcommand{\Lie}{\mathrm{Lie}}

\newcommand{\FitZero}{\mathrm{F}_0}

\newcommand{\too}%
{\xrightarrow{\text{\raisebox{-3pt}{$\sim$}}\,}}